\documentclass[12pt,reqno]{amsart}
\usepackage[top=1.2in, bottom=1.2in, left=1.2in, right=1.2in]{geometry}
\usepackage{amsfonts, amssymb, amsmath, mathtools, dsfont, xcolor}
\usepackage[linktoc=page, colorlinks, linkcolor=blue, citecolor=blue]{hyperref}
\usepackage{enumerate, commath, MnSymbol}
\usepackage{algorithm,algpseudocode}

\usepackage{graphicx}

\usepackage{subcaption} 
\usepackage[font=small]{caption} 
\numberwithin{equation}{section}

\DeclareMathOperator{\E}{\mathbb{E}}

\DeclareMathOperator{\Var}{Var}

\DeclareMathOperator*{\argmin}{argmin}
\DeclareMathOperator{\Lap}{Lap}
\DeclareMathOperator{\range}{range}

\renewcommand{\Pr}[2][]{\mathbb{P}_{#1} \left\{ #2 \rule{0mm}{3mm}\right\}}

\newcommand{\iip}[2]{\llangle#1,#2\rrangle}

\def \R {\mathbb{R}}

\def \AA {\mathcal{A}}

\def \FF {\mathcal{F}}

\def \LL {\mathcal{L}}

\def \g {\gamma}
\def \e {\varepsilon}
\def \d {\delta}
\def \l {\lambda}
\def \s {\sigma}

\def \one {{\textbf 1}}

\newtheorem{theorem}{Theorem}[section]

\newtheorem{corollary}[theorem]{Corollary}
\newtheorem{lemma}[theorem]{Lemma}

\newtheorem{definition}[theorem]{Definition}

\theoremstyle{remark}

\begin{document}

\title{Privacy of synthetic data: a statistical framework}
\author{March Boedihardjo}
\address{Department of Mathematics, University of California Irvine}
\email{marchb@math.uci.edu}
\author{Thomas Strohmer}
\address{Center of Data Science and Artificial Intelligence Research, University of California, Davis \\ \& Department of Mathematics, University of California Davis}
\email{strohmer@math.ucdavis.edu}
\author{Roman Vershynin}
\address{Department of Mathematics, University of California Irvine}
\email{rvershyn@uci.edu}

\maketitle

\begin{abstract}
Privacy-preserving data analysis is emerging as a challenging problem with far-reaching impact. In particular, synthetic data are a promising concept toward solving the aporetic conflict between data privacy and data sharing. Yet, it is known that accurately generating private,  synthetic data of certain kinds is NP-hard.  We develop a statistical framework for differentially private synthetic data, which enables us to circumvent the computational hardness of the problem. We consider the true data as a random sample drawn from a population $\Omega$ according to some
unknown density. We then replace $\Omega$ by a much smaller random subset $\Omega^*$, which we sample according to some known density. We generate synthetic data
on the reduced space $\Omega^*$ by fitting the specified linear statistics obtained from the true data.  To ensure privacy we use the common Laplacian mechanism.  Employing the concept of  R\`enyi condition number, which measures how well the sampling distribution is correlated with the population distribution, we derive explicit bounds on the privacy and accuracy provided by the proposed method.

\end{abstract}


\section{Introduction}

Data science and artificial intelligence play a key role in successfully tackling many of the grand challenges our
society is facing over the coming years.  Data sharing and data democratization will feature prominently in these endeavors.
At the same time, data colonialism~\cite{couldry2019data} and surveillance capitalism~\cite{zuboff2019} emerge as increasingly concerning developments that threaten the potential benefits of data-driven advancements and that highlight the utmost importance of 
data rights and privacy. For instance, the WHO emphasized in its recent report the importance of data management methods that improve the utility and accuracy of health-care data, while not compromising privacy~\cite{WHO2021}.
 However, data democratization and responsible data sharing are not likely to be accommodated by more efficient deidentification or strict security/privacy processes alone.

{\em Synthetic data} is a promising ingredient toward solving the aporetic conflict between data privacy and data sharing. The goal of synthetic data is  to create an as-realistic-as-possible data set, one that not only maintains the nuances of the original data, but does so without risk of exposing sensitive information.
The problem of making private and accurate synthetic data is NP-hard in the worst case~\cite{ullman2011pcps,ullman2016answering}.

In this paper we take a different route. We will show that the problem of making private and accurate synthetic data
is tractable in the statistical framework, where the true data is seen as a random sample drawn from some probability space.
Our method comes with guarantees of privacy, accuracy, and computational efficiency.
We will discuss how our method improves upon existing techniques in Section~\ref{stateoftheart}.


\section{Problem setup and main results}

\subsection{The problem}
We model the true data $X$ as a sequence of $n$ elements  from some ground set $\Omega$. E.g., for an electronic health record these elements might represent patients.
For example, $\Omega = \{0,1\}^p$ allows each patient to have $p$ binary parameters, while 
$\Omega = \R^p$ allows the parameters to be real. Multimodal data are possible, too: some parameters may be categorical, some real, some may consist of text strings, etc. 
We would like to manufacture a synthetic dataset $Y$, which is another sequence of $k$ elements from $\Omega$. We want the synthetic data to be  private and accurate.

\subsection{Defining accuracy}
By ``accuracy'' we mean the accuracy of linear statistics of the data. 
Consider a finite class $\FF$ of {\em test functions}, which are functions from $\Omega$ to $[-1,1]$. Linear statistics of the data $X = (x_1,\ldots,x_n)$ are the sums of the form
$\frac{1}{n} \sum_{i=1}^n f(x_i)$ for $f \in \FF$.
We would like the synthetic data $Y$ to approximately preserve all these sums, up to a given additive error $\d$:
\begin{equation}	\label{eq: delta-accuracy}
\max_{f \in \FF} 
\abs{\frac{1}{k} \sum_{i=1}^k f(y_i) - \frac{1}{n} \sum_{i=1}^n f(x_i)} \le \d.
\end{equation}
In this case we say that the synthetic dataset is {\em $\d$-accurate}.

As an important example, linear statistics are capable of encoding {\em marginals} of high-dimensional data. Indeed, let us consider Boolean data where $\Omega = \{0,1\}^p$. In the context of electronic health records, the data $X = (x_1,\ldots,x_n)$ consists of records of $n$ patients each having $p$ binary parameters. The fraction of the number of patients whose first and second parameters equal $1$ and third parameter equals $0$ is a three-dimensional marginal. It can be expressed as the linear statistic
$\frac{1}{n} \sum_{i=1}^n f(x_i)$, where $f : \{0,1\}^p \to \{0,1\}$ is the indicator function 
$f(x) = \one_{\{x(1)=x(2)=1, \, x(3)=0\}}$. One-dimensional marginals capture the means of the parameters, jointly with two-dimensional marginals they determine the correlations, 
and higher dimensional marginals capture higher-order dependencies.

In many situations, $\abs{\Omega}$ is too large for computations while $\abs{\FF}$ is reasonable.
For example, if $\FF$ encodes all $d$-dimensional marginals of $p$-dimensional Boolean data as in the previous example, $\abs{\Omega} = 2^p$ is exponential in $p$, while 
$$
\abs{\FF} = \binom{p}{\le d} = \binom{p}{0}+\binom{p}{1}+\cdots + \binom{p}{d} 
\le \Big(\frac{ep}{d}\Big)^d
$$
is polynomial in $p$ for any fixed $d$. 

\subsection{A statistical framework}

Ullman and Vadhan~\cite{ullman2011pcps} showed (under standard cryptographic assumptions) that in general it is NP-hard to make private synthetic Boolean data which approximately preserve all two-dimensional marginals. While this result may seem discouraging, it is a worst-case result. 

Yet the {\em worst} kind of data, for which the problem is hard, are rarely seen in practice. More common in applications is the statistical framework, where the true data is seen as a {\em random} sample drawn from some probability space $(\Omega, \Sigma, \nu)$. 
The probability distribution $\nu$ specifies the population model of the true data.
We assume that we neither know $\nu$, 
nor can we sample according to $\nu$ thereby generating more true data.

Suppose, however, that we can sample from $\Omega$ according to some other, known, probability measure $\mu$. For example, while we may not know the underlying population distribution $\nu$ of the patients in the Boolean cube $\Omega = \{0,1\}^p$, we can still sample from the cube according to the uniform measure $\mu$ by choosing all coordinates at random and independently. Similarly, while we may not know the population distribution $\nu$ of written notes in patient health records, there do exist generative models $\mu$ that generate texts. In order to uphold privacy, we assume that the true data $X$ may not be used to build the generative model $\mu$, but it can be built using some other public data. 

Having put our problem into a statistical framework, we can try to circumvent the computational hardness of our problem in the most obvious way: {\em subsample} $\Omega$. Namely, we replace $\Omega$ by a much smaller random subset $\Omega^*$ that is sampled according to the distribution $\mu$. Then we generate synthetic data in $\Omega^*$ by fitting the desired linear statistics (e.g.\ all marginals up to a specified degree) of the true data as close as possible. 

This idea may only work if the sampling distribution $\mu$ has some ``correlation'' with the population distribution $\nu$. We can quantify this correlation using the notion of {\em R\`enyi divergence}~\cite{renyi1961measures}. Namely, if $\nu$ is absolutely continuous with respect to $\mu$, we can utilize the Radon-Nikodym derivative $d\nu/d\mu$ to define the {\em R\`enyi condition number}
\begin{equation}\label{kappa}
\kappa(\nu\|\mu) 
= \int \Big( \frac{d\nu}{d\mu} \Big)^2 d\mu
= \int \frac{d\nu}{d\mu} \, d\nu,
\end{equation}
a quantity that equals the exponential of $D_2(\nu\|\mu)$, 
the {\em R\`enyi divergence} of order $2$.

Conceptually, $\kappa(\nu\|\mu)$ is similar to the notion of the condition number in numerical linear algebra: the smaller, the better. The best value of the R\`enyi condition number is $1$, achieved when $\nu=\mu$.

If $\Omega$ is finite, the Radon-Nikodym derivative
$d\nu/d\mu$ equals the ratio of the densities $\phi(x)=\nu(\{x\})$ and $\psi(x)=\mu(\{x\})$.
In particular, if the sampling distribution $\mu$ is uniform, $\psi(x) = 1/\abs{\Omega}$ for all $x$, 
and we have
\begin{equation}	\label{eq: condition number uniform}
\kappa(\nu\|\mu) 
= \int \phi(x)^2 \abs{\Omega}^2 \, d\mu(x)
= \Bigg( \frac{\norm{\phi}_{L^2(\mu)}}{\norm{\phi}_{L^1(\mu)}} \Bigg)^2.
\end{equation}
Thus, the R\`enyi condition number in this case measures the regularity of the population density $\phi$: 
the more spread out it is, the smaller its R\`enyi condition number.

\subsection{Our approach}

Our method, in a nutshell, is the following:
obtain a reduced space $\Omega^*$ by subsampling $\Omega$ according to the known probability measure $\mu$, and generate synthetic data $Y$ on $\Omega^*$ by fitting the linear statistics obtained from $X$. 

Our results come with guarantees of privacy, accuracy, and efficiency. To achieve all this, we assume (roughly speaking) that the size of the true data is at least nearly linear in the number of statistics we seek to preserve: 
$$
\abs{X} \gtrsim \abs{\FF} \log \abs{\FF}.
$$
For accuracy, we need the size of the synthetic data to be at least logarithmic in the number of statistics (a mild assumption):
$$
\abs{Y} \gtrsim \log \abs{\FF}.
$$
And, finally, we can make all computations in the reduced space $\Omega^*$ as long as its size is at least linear in the number of statistics:
$$
\abs{\Omega^*} \gtrsim \abs{\FF}.
$$
If these three conditions are met, we can generate synthetic data while preserving privacy, accuracy, and efficiency (for the latter, we solve a linear program in dimension  $\abs{\Omega^*}$).

In order to provide rigorous privacy guarantees, we will employ the concept of {\em differential privacy}~\cite{dwork2014algorithmic}, which has emerged as a de-facto standard for private data sharing. 
\begin{definition}[Differential Privacy~\cite{dwork2014algorithmic}]  A randomized function ${\mathcal M}$ gives $\epsilon$-differential privacy
 if for all databases $D_1$ and $D_2$  differing on at most one element, and all measurable $S \subseteq \range({\mathcal M})$, 
 $${\mathbb P} [{\mathcal M}(D_1) \in S] \le e^{\epsilon} \cdot {\mathbb P} [{\mathcal M}(D_2) \in S],$$
where the probability is with respect to the randomness of  ${\mathcal M}$.
\end{definition}

A basic technique to achieve differential privacy is the {\em Laplacian mechanism}, which consists of adding Laplacian noise to the data. A Laplacian random variable $\l$ is Laplacian with parameter $\s$, abbreviated $\l \sim \Lap(\s)$, if $\l$ is a symmetric random variable with exponential tails in both directions:
$$
\Pr{\abs{\l} > t} = \exp(-t/\s), \quad t \ge 0.
$$
It is well known and not hard to see that Laplacian mechanism achieves differential privacy;
see Lemma~\ref{lem: Laplacian mechanism} for details.

\subsection{Algorithm}

We present a high level algorithmic description of our proposed method in Algorithm~\ref{mainalgorithm} below. See Section~\ref{ss:privacy} for the role of the parameters arising in the algorithm.

\begin{algorithm}[h!]
\caption{Private synthetic data algorithm}
\label{mainalgorithm}
\begin{algorithmic}

\State {\bf Input:}
(a) the true data: a sequence $X=(x_1,\ldots,x_n)$ of $n$ elements of $\Omega$;\\
\hspace*{13mm} (b) a family $\FF$ of test functions from $\Omega$ to $[-1,1]$;\\
\hspace*{13mm} (c) the reduced space $\Omega^* = \{z_1,\ldots,z_m\}$, made of
points $z_i$ chosen from $\Omega$;
\hspace*{13mm} (d) parameter $\s > 0$.

\begin{enumerate}
\item[\bf {1.}]
{\bf Add noise:} For each test function $f \in \FF$, generate an independent Laplacian random variable
$\l(f) \sim \Lap(\s)$.

\item[\bf {2.}]
{\bf Reweight:} Compute a density $h^*$ on $\Omega^*$ whose linear statistics are uniformly as close as possible to the linear statistics of the true data perturbed by  Laplacian noise:
$$
h^* = \argmin \left\{ \max_{f \in \FF} \abs{\sum_{i=1}^m f(z_i) h(z_i) - \frac{1}{n} \sum_{i=1}^n f(x_i) - \l(f)}: \;
h \text{ is a density on } \Omega^* \right\}.
$$

\item[\bf{3.}]
{\bf Bootstrap:} Create a sequence $Y = (y_1,\ldots,y_k)$ of $k$ elements 
drawn from $\Omega^*$ independently with density $h^*$.

\end{enumerate}

\State {\bf Output:} synthetic data $Y = (y_1,\ldots,y_k)$.

\end{algorithmic}

\end{algorithm}

\bigskip

Note that computing $h^*$ amounts to solving a linear program with $\abs{\Omega^*} \le m$ variables\footnote{We have inequality here because the set $\Omega^*$ is formed of points $z_i$ that are sampled independently, which may result in repetitions.} and at most 
$\abs{\FF} +m+1$ constraints. The complexity of solving general linear programs is polynomial in the number of variables, see e.g.~\cite{megiddo2012progress}.

\subsection{Privacy and accuracy guarantees} \label{ss:privacy}

\begin{theorem}[Privacy]		\label{thm: privacy}
Let $\d >0, \g >0$ and set $\s = \d / \log(\abs{\FF}/\g)$. If $$n \ge 2(\e\d)^{-1} \abs{\FF} \log(\abs{\FF}/\g),$$  then Algorithm~\ref{mainalgorithm} is $\e$-differentially private.
\end{theorem}

We emphasize that this privacy guarantee holds for {\em any} choice of the reduced space $\Omega^*$. 

\begin{theorem}[Accuracy]		\label{thm: accuracy}
  Let $\min(n,k) \ge \d^{-2} \log(\abs{\FF}/\g)$ and $m \ge \d^{-2} K \abs{\FF}/\g$, where  $\d \in (0,1/2]$ and $\g \in (0,1/4)$. Set 
  $\s = \d / \log(\abs{\FF}/\g)$.
  Suppose the true data $X = (x_1,\ldots,x_n)$ is sampled from $\Omega$ 
  independently and according to some probability measure $\nu$, 
  and the reduced space $\Omega^* = \{z_1,\ldots,z_m\}$ is sampled
  from $\Omega$ independently and according to some probability measure $\mu$.
  Assume that the R\`enyi condition number satisfies $\kappa(\nu\|\mu) \le K$.
  Also assume that the family $\FF$ contains the function that is identically equal to $1$.
  Then with probability at least $1-4\g$ 
  the synthetic data $Y = (y_1,\ldots,y_k)$ generated by  Algorithm~\ref{mainalgorithm} is $(8\d)$-accurate.
 \end{theorem}

Let us specialize our results to Boolean data. 
Here the sample space is $\Omega = \{0,1\}^p$ 
and we seek accuracy with respect to all $\abs{\FF} = \binom{p}{\le d}$ marginals up to degree $d$. Choose $\mu$ to be the uniform density on the cube, recall \eqref{eq: condition number uniform}, and combine the two theorems above to get:

\begin{corollary}[Boolean data]		\label{cor: Boolean data}
  Let $n \gg \binom{p}{\le d} \log \binom{p}{\le d}$ 
  and $k \gg \log \binom{p}{\le d}$.
  Suppose that the true data $X = (x_1,\ldots,x_n)$ is sampled from $\{0,1\}^p$ 
  independently and according to some (unknown) density $\phi$.
  Then one can generate synthetic data $Y= (y_1,\ldots,y_k)$ that is
  $o(1)$-accurate with respect to all marginals 
  of dimension at most $d$ with probability $1-o(1)$, and is also $o(1)$-differentially private.
  The algorithm that generates $Y$ from $X$ runs in time polynomial in 
  $n$, $k$, and $\kappa$ for a fixed $d$.  
  \end{corollary} 

The proofs of the claims above will be given in Section~\ref{s:proofs}.

\subsection{Related work}\label{stateoftheart}

There exists a fairly large body of work on privately releasing answers in the interactive and non-interactive query setting, a detailed review of which is beyond the scope of this paper.
A major advantage of releasing a  synthetic data set instead of just the answers to specific queries is that synthetic data opens up  a much richer toolbox (clustering, classification, regression, visualization, etc.), and thus much more flexibility, to analyze the data.

In~\cite{blum2013learning}, Blum, Ligett, and Roth gave an $\e$-differentially private synthetic data algorithm whose accuracy scales  logarithmically with the number of queries, but the complexity scales exponentially with $p$. This computational inefficiency comes as no surprise, if we recall that making differentially private Boolean synthetic data which preserves all of the
two-dimensional marginals with accuracy $o(1)$ is NP-hard~\cite{ullman2011pcps}.

The papers~\cite{hardt2010multiplicative,hardt2012simple} propose methods for producing private synthetic data with an error bound of about  $\tilde{\mathcal{O}}(\sqrt{n} p^{1/4})$ per query. However, the associated  algorithms have running time
that is at least exponential in $p$.

In~\cite{barak2007privacy}, Barak et al.\ derive a method for producing accurate and private synthetic Boolean data based on linear programming. The method in~\cite{barak2007privacy} is conceptually similar to ours even though it is concerned with marginals, while our approach holds for general linear statistics. The key difference is in the computational complexity. The method in~\cite{barak2007privacy} involves solving a linear program on the entire domain $\Omega = \{0,1\}^p$ and thus its running time is exponential in $p$. The authors of~\cite{barak2007privacy} emphasize that {\em ``one of the main algorithmic questions left open from this work is that of efficiency''}, for which our paper provides a solution.
Our method works in the reduced space $\Omega^*$, which, according to Theorem~\ref{thm: accuracy}, has size $m$ slightly larger than $\binom{p}{\le d}$, and thus it is only polynomial in $p$, thereby providing a positive answer to the aforementioned algorithmic question.

The method developed by Hardt and Talwar in \cite{hardttalwar} privately releases answers to linear queries (including, in particular, marginals). It applies to general data that needs not be Boolean, just like in our work. However, unlike our method, the method in \cite{hardttalwar} does not construct synthetic data. Also, unlike our work, the theoretical accuracy bounds in \cite{hardttalwar} hold for most but not all linear queries.
Nikolov, Talwar, and Zhang in~\cite{nikolov2013geometry}, follow up on the work~\cite{hardttalwar} and improve the (lower and upper) bounds derived by Hardt and Talwar.  The lack of efficiency of the method in~\cite{nikolov2013geometry} is addressed in~\cite{aydore2021differentially}, where the authors demonstrate empirically the computational efficiency of their method.

The paper~\cite{dwork2015efficient} by Dwork, Nikolov, and Talwar is concerned with a convex relaxation based approach for private marginal release, and thus, unlike our method,  does not construct synthetic data for a ground set  $\Omega$.  Also, \cite{dwork2015efficient} gives ``only'' $(\epsilon,\delta)$-differential privacy.

Privacy-preserving data analysis (beyond marginals) in a statistical framework is the focus of~\cite{duchi2018minimax,cai2019cost}. While these papers are quite intriguing, they are not concerned with synthetic data, and thus not directly related to this work.

Another method of constructing private synthetic data was proposed recently in  \cite{BSV2021a}. To compare the two, recall that the no-go result of Ullhman says (roughly) that, for the {\em worst true data}, it is impossible to efficiently construct private synthetic Boolean data that approximately preserves {\em all marginals} of dimension $2$. The work \cite{BSV2021a} and the present paper overcome this impossibility result, each in its own way: this paper relaxes ``worst data'' to ``typical data'', while \cite{BSV2021a} relaxes ``all marginals'' to ``most marginals''.

\section{Proofs}\label{s:proofs}

For an integrable function $f: \Omega \to \R$ on a measure space $(\Omega, \Sigma, \nu)$, 
we denote
\begin{equation}	\label{eq: bilinear form}
\iip{f}{\nu} = \int f\, d\nu.
\end{equation}
Given a sequence of points $x_1,\ldots,x_n \in \Omega$, possibly with repetitions, 
we consider the empirical measure 
$$
\nu_n = \frac{1}{n} \sum_{i=1}^n \d_{x_i}.
$$
By definition, we have
\begin{equation}	\label{eq: bilinear empirical}
\iip{f}{\nu_n} = \frac{1}{n} \sum_{i=1}^n f(x_i).
\end{equation}
With this notation, the optimization part of Algorithm~\ref{mainalgorithm} can be expressed as follows:
\begin{equation}	\label{eq: algorithm equivalent}
h^* = \argmin \left\{ \max_{f \in \FF} \abs{\iip{f}{h} - \iip{f}{\nu_n}- \l(f)}: \;
h \text{ is a probability measure on } \Omega^*
\right\}.
\end{equation}

\subsection{Privacy}

The following lemma is well known, see e.g.\ Theorem~2 in~\cite{barak2007privacy}.

\begin{lemma}[Laplacian mechanism]			\label{lem: Laplacian mechanism}
  Let $\AA$ be a mapping that transforms data $D$ to a point $\AA(D) \in \R^N$.
  Let   
  $$
  \Delta = \max_{D_1, D_2} \norm{\AA(D_1)-\AA(D_2)}_1
  $$
  where the maximum is over all pairs of input data $D_1$ and $D_2$ 
  that differ in a single element. 
  Then the addition of i.i.d.\ Laplacian noise $\l_i \sim \Lap(\s)$ to each coordinate
  of $\AA(D)$ preserves $(\Delta/\s)$-differential privacy.
\end{lemma}

Consider the linear map $\LL$ that associates to a measure $\nu$ on $\Omega$
the set of its linear statistics, namely
$$
\LL(\nu) = \left( \iip{f}{\nu} \right)_{f \in \FF} \in \R^{\abs{\FF}}.
$$
Consider two input sets $(x_1,\ldots,x_n)$ and $(x_1,\ldots,x_n,x_{n+1})$
that differ by exactly one element $x_{n+1}$. 
Then one can easily check that the corresponding empirical measures satisfy the identity
$$
\nu_{n+1} - \nu_n = \frac{1}{n+1} \left( \d_{x_{n+1}} - \nu_n \right).
$$
Then, using linearity of $\LL$ and the triangle inequality, we obtain
\begin{equation}	\label{eq: L stability}
\norm{\LL(\nu_{n+1})-\LL(\nu_n)}_1
= \norm{\LL(\nu_{n+1}-\nu_n)}_1
\le \frac{1}{n+1} \norm{\LL(\d_{x_{n+1}})}_1 + \frac{1}{n+1} \norm{\LL(\nu_n)}_1.
\end{equation}
To bound this quantity further, note that for every $i$ the definition of $\LL$ yields
\begin{equation}	\label{eq: L delta}
\norm{\LL(\d_{x_i})}_1
= \sum_{f \in \FF} \abs{\iip{f}{\d_{x_i}}}
= \sum_{f \in \FF} \abs{f(x_i)}
\le \abs{\FF},
\end{equation}
where in the last step we used that each function $f \in \FF$ takes values in $[-1,1]$.
Therefore, by linearity of $\LL$ and the triangle inequality, 
$$
\norm{\LL(\nu_n)}_1 
= \norm[3]{\frac{1}{n} \sum_{i=1}^n \LL(\d_{x_i})}_1
\le \frac{1}{n} \sum_{i=1}^n \norm{\LL(\d_{x_i})}_1
\le \abs{\FF},
$$
where in the last step we used \eqref{eq: L delta}.
Substituting the bound \eqref{eq: L delta} for $i=n+1$ and the last inequality into 
\eqref{eq: L stability}, we conclude that 
$$
 \Delta \coloneqq \norm{\LL(\nu_{n+1})-\LL(\nu_n)}_1
\le \frac{2\abs{\FF}}{n}.
$$

Applying Lemma~\ref{lem: Laplacian mechanism}, we see that the addition of the independent 
Laplacian random variable $\l(f) \sim \Lap(\s)$ to each coordinate $\iip{f}{\nu_n}$ of $\LL(\nu_n)$ preserves $(\Delta/\s)$-differential privacy. Due to the bound on $\Delta$ above, the choice of $\s$ in the algorithm, and the assumption on $n$ in Theorem~\ref{thm: privacy}, we have
$$
\frac{\Delta}{\s} 
\le \frac{2\abs{\FF} \log(\abs{\FF}/\g)}{n\d}
\le \e.
$$
Hence, the family of perturbed coefficients $\iip{f}{\nu_n}+\l(f)$ is $\e$-differentially private. Finally, the function  $h^*$ in~\eqref{eq: algorithm equivalent} computed by the algorithm is a function of these private perturbed coefficients. Hence the algorithm is $\e$-differentially private. Theorem~\ref{thm: privacy} is proved.

\subsection{Accuracy}

Here, our input data $X_1,\ldots,X_n$ are i.i.d.\  points sampled from $\Omega$ 
according to the probability measure $\nu$, and the reduced space $\Omega^*$ is formed by the points $Z_1,\ldots,Z_m$ sampled from $\Omega$ 
according to the probability measure $\mu$. 
Consider the corresponding empirical probability measures
$$
\nu_n = \frac{1}{n} \sum_{i=1}^n \d_{X_i}
\quad \text{and} \quad
\mu_m = \frac{1}{m} \sum_{i=1}^m \d_{Z_i}.
$$
Let us reweigh the reduced space, introducing the measure
\begin{equation}	\label{eq: reweighed}
\nu'_m = \frac{1}{m} \sum_{i=1}^m \Big( \frac{d\nu}{d\mu}\Big)(Z_i) \, \d_{Z_i}.
\end{equation}
The point is that both $\nu_n$ and $\nu'_m$ are unbiased estimators of the population measure $\nu$:
$$
\E \nu_n = \E \nu'_m = \nu.
$$
These identities can be easily deduced from the definition of the Radon-Nikodym derivative. 
In our argument, however, they will not be used. Instead, we need uniform deviation inequalities that would guarantee that with high probability, all linear statistics of $\nu_n$, $\nu'_m$ and $\nu$ approximately match. This is the content of the next two lemmas.

\begin{lemma}[Deviation of linear statistics for $\nu_n$]		\label{lem: deviation nun}
  Let $(\Omega, \Sigma, \nu)$ be a probability space, and let
  $\nu_n$ be an empirical probability measure corresponding to $\nu$.
  If $n \ge \d^{-2} \log(\abs{\FF}/\g)$ then, 
  with probability at least $1-\g$, we have
  $$
  \max_{f \in \FF} \abs{\iip{f}{\nu_n} - \iip{f}{\nu}} \le \d.
  $$
\end{lemma}

\begin{proof}
For each function $f \in \FF$, recalling \eqref{eq: bilinear form} and \eqref{eq: bilinear empirical} we get 
$$
\iip{f}{\nu} 
= \int f \, d\nu
= \E f(X), \quad 
\iip{f}{\nu_n} 
= \frac{1}{n} \sum_{i=1}^n f(X_i),
$$
where $X, X_1, X_2,\ldots$ are drawn from $\Omega$ independently according to probability measure $\nu$.
Therefore 
$$
\iip{f}{\nu_n} - \iip{f}{\nu}
= \frac{1}{n} \sum_{i=1}^n \left( f(X_i) - \E f(X_i) \right)
$$
is a normalized and centered sum of i.i.d.\ random variables, which are 
bounded by $1$ in absolute value (by assumption on $\FF$).
Applying Bernstein's inequality (see e.g. \cite[Theorem~2.8.4]{vershyninbook}) we get for any $\d \in (0,1)$ that
$$
\Pr{\abs{\iip{f}{\nu_n} - \iip{f}{\nu}}>\d}
\le \exp(-\d^2 n)
\le \g/\abs{\FF},
$$
where in the last step we used the assumption on $n$.
The lemma is proved.
\end{proof}

\begin{lemma}[Deviation of linear statistics for $\nu'_m$]	\label{lem: deviation reweighed}
  If $m \ge \d^{-2} K \abs{\FF}/\g$ and $\kappa(\nu\|\mu) \le K$ then, 
  with probability at least $1-\g$, we have
  $$
  \max_{f \in \FF} \abs{\iip{f}{\nu'_m} - \iip{f}{\nu}} \le \d.
  $$
\end{lemma}

\begin{proof}
For each test function $f \in \FF$, by definition of the Radon-Nikodym derivative, we have
$$
\iip{f}{\nu}
= \int f \, d\nu
= \int f(z) \Big( \frac{d\nu}{d\mu} \Big)(z) \, d\mu(z)
= \E \Big( \frac{d\nu}{d\mu} \Big)(Z) f(Z),
$$
where $Z$ is drawn from $\Omega$ according to probability measure $\mu$.
Furthermore, by definition of reweighting \eqref{eq: reweighed} we have
$$
\iip{f}{\nu'_m}
= \int f \, d\nu'_m
= \frac{1}{m} \sum_{i=1}^m \Big( \frac{d\nu}{d\mu}\Big)(Z_i) \, f(Z_i),
$$ 
where $Z_i$ are i.i.d.\ copies of $Z$.
Therefore 
$$
\iip{f}{\nu'_m} - \iip{f}{\nu}
= \frac{1}{m} \sum_{i=1}^m \left( R_i - \E R_i \right)
\quad \text{where} \quad 
R_i = \Big( \frac{d\nu}{d\mu}\Big)(Z_i) \, f(Z_i).
$$
In other words, we have a normalized and centered sum of i.i.d.\ random variables. The variance of each term of the sum is bounded by the R\`enyi condition number $\kappa(\nu\|\mu)$. Indeed,  
$$
\Var(R_i)
\le \E R_1^2
= \E \Big( \frac{d\nu}{d\mu}\Big)(Z)^2 \, f(Z)^2
\le \E \Big( \frac{d\nu}{d\mu}\Big)(Z)^2
= \int \Big( \frac{d\nu}{d\mu}\Big)^2 \, d\mu
= \kappa(\nu\|\mu) \le K.
$$
Here, in the third step we used the assumption that $f$ takes values in $[-1,1]$.

We showed that the variance of $\iip{f}{\nu'_m} - \iip{f}{\nu}$ is bounded by $K/m$.
Applying Chebyshev's inequality, we get for any $\d \in (0,1)$ that
$$
\Pr{\abs{\iip{f}{\nu'_m} - \iip{f}{\nu}}>\d}
\le \frac{K}{\d^2 m}
\le \frac{\g}{\abs{\FF}},
$$
where in the last step we used the assumption on $m$.
The lemma is proved.
\end{proof}

\begin{proof}[Proof of Theorem~\ref{thm: accuracy}]
Assume that the events in the conclusions of Lemma~\ref{lem: deviation nun} and Lemma~\ref{lem: deviation reweighed} hold; this happens with probability at least $1-2\gamma$.

The measure $\nu'_m$ introduced in \eqref{eq: reweighed} need not be a probability measure, 
since its total mass 
$$
r \coloneqq \iip{\one}{\nu'_m} 
$$
does  not need to equal $1$. 
But it is not far from $1$. 
Indeed, since the constant function $\one$ lies in $\FF$ by assumption, 
the conclusion of Lemma~\ref{lem: deviation reweighed} gives
$$
\abs{\iip{\one}{\nu'_m} - \iip{\one}{\nu}} \le \d.
$$
Since $\nu$ is a probability measure, it satisfies $\iip{\one}{\nu}=1$, and we get
\begin{equation}	\label{eq: r}
\abs{r-1} \le \d.
\end{equation}

Now, $\nu'_m/r$ is a probability measure. Let us check that it satisfies a deviation inequality.
To this end, first note that the conclusion of Lemma~\ref{lem: deviation reweighed} and triangle inequality give
\begin{equation}	\label{eq: abs class}
\abs{\iip{f}{\nu'_m}} 
\le \abs{\iip{f}{\nu}} + \d
= \abs{\int f \, d\nu} + \d
\le 1+\d
\end{equation}
where we used the assumption that all $f \in \FF$ take values in $[-1,1]$.
Thus, subtracting and adding the term $\iip{f}{\nu'_m}$, we obtain
$$
\abs{ \iip{f}{\nu'_m/r} - \iip{f}{\nu} } 
\le \abs{1/r-1} \, \abs{\iip{f}{\nu'_m}} + \abs{ \iip{f}{\nu'_m} - \iip{f}{\nu} }.
$$
Since $\d \in (0,1/2]$,
\eqref{eq: r} yields $\abs{1/r-1} \le 2\d$. Furthermore, \eqref{eq: abs class} yields $\abs{\iip{f}{\nu'_m}} \le 3/2$. Finally, the conclusion of Lemma~\ref{lem: deviation reweighed} yields $\abs{ \iip{f}{\nu'_m} - \iip{f}{\nu} } \le \d$. Substituting these bounds into the inequality above, we obtain the desired deviation inequality:
$$
\max_{f \in \FF} \abs{ \iip{f}{\nu'_m/r} - \iip{f}{\nu} } \le 4\d.
$$

Combining this with the conclusion of Lemma~\ref{lem: deviation nun} via the triangle inequality, we obtain 
$$
\max_{f \in \FF} \abs{ \iip{f}{\nu'_m/r} - \iip{f}{\nu_n} } \le 5\d.
$$
A simple union bound over $\abs{\FF}$ Laplacian random variables shows that
with probability at least $1-\g$,
\begin{equation}	\label{eq: Laplacian bound}
\max_{f \in \FF} \abs{\l(f)} \le \s \log(\abs{\FF}/\g) = \d
\end{equation}
where the last identity is due to the choice of $\s$ in the algorithm.
Combining the two bounds, with probability at least $1-3\g$, we have
$$
\max_{f \in \FF} \abs{ \iip{f}{\nu'_m/r} - \iip{f}{\nu_n} - \l(f) } \le 6\d.
$$

Recall that, by construction, $\nu'_m/r$ is a probability measure on the set $\Omega^* = \{Z_1,\ldots,Z_m\}$.
Therefore, minimality of $h^*$ in algorithm \eqref{eq: algorithm equivalent} implies that 
$$
\max_{f \in \FF} \abs{ \iip{f}{h^*} - \iip{f}{\nu_n} - \l(f) } \le 6\d.
$$
Using \eqref{eq: Laplacian bound} again, we conclude that 
$$
\max_{f \in \FF} \abs{ \iip{f}{h^*} - \iip{f}{\nu_n} } \le 7\d.
$$

To complete the proof, we note that bootstrapping preserves the accuracy of linear statistics. Indeed, apply Lemma~\ref{lem: deviation nun} for the probability density $h^*$ on $\Omega^*$ and its empirical counterpart $h^*_k = \frac{1}{k} \sum_{i=1}^k \d_{Y_i}$ where $Y_i$ are sampled independently from $\Omega^*$ according to the probability density $h^*$.
Since $k \ge \d^{-2} \log(\abs{\FF}/\g)$ by assumption, 
with probability at least $1-\g$  we have 
$$
\max_{f \in \FF} \abs{\iip{f}{h^*_k} - \iip{f}{h^*}} \le \d.
$$
Combining this with the previous bound, we obtain that with probability at least $1-4\g$, 
$$
\max_{f \in \FF} \abs{\iip{f}{h^*_k} - \iip{f}{\nu_n}} \le 8\d.
$$
This is an equivalent form of $(8\d)$-accuracy \eqref{eq: delta-accuracy}.
Theorem~\ref{thm: accuracy} is proved.
\end{proof}

\section{Open problems}

While the method proposed in this paper provides a simple and efficient roadmap to construct private synthetic data that preserve with high accuracy linear statistics of the original data, we may require our synthetic data to accurately model other features of the data that are not (fully) captured by linear statistics. This poses numerous questions. For example, how well do linear statistics inform other kinds of  data analysis tasks (e.g., clustering, classification, regression, etc.)?  

Another challenge is that we do not know the population distribution $\nu$, and thus we may not know how to choose a good sampling distribution $\mu$. Using various generative models seem a natural choice for certain types of data, such as text and images. Using those, we may hope to build the sampling distribution $\mu$ that has enough ``overlap'' with the population distribution $\nu$ (as measured by the Renyi condition number). Since we just need to be able to sample from $\nu$, building an MCMC model for it is enough. 

It is important, however, that we may not use the true data $X$ to make any decisions about $\mu$, as this could
violate privacy. The sampling distribution $\mu$ should be estimated in some other way. We can either use private
density estimation for that purpose, or estimate $\mu$ from some publicly available data that does not need to be protected by privacy.

\section*{Acknowledgement}

M.B. acknowledges support from NSF DMS-2140592. T.S. acknowledges support from NSF-DMS-1737943, NSF DMS-2027248, NSF CCF-1934568 and a CeDAR Seed grant.
 R.V. acknowledges support from NSF DMS-1954233, NSF DMS-2027299, U.S. Army 76649-CS, and NSF+Simons Research Collaborations on the Mathematical and Scientific Foundations of Deep Learning.

\end{document}